\documentclass{article}

\usepackage{arxiv}

\usepackage[utf8]{inputenc} 
\usepackage[T1]{fontenc}    
\usepackage{hyperref}       
\usepackage{url}            
\usepackage{booktabs}       
\usepackage{amsfonts}       
\usepackage{nicefrac}       
\usepackage{microtype}      

\usepackage{amsmath}
\usepackage{siunitx}
\usepackage{amssymb}
\usepackage{comment}
\usepackage{amsthm}
\usepackage{graphicx}
\usepackage{subfig}
\usepackage{comment}
\usepackage{nidanfloat}
\usepackage{bm}
\usepackage{hyperref}
\usepackage{tabularx} 

\newtheorem{theorem}{Theorem}

\newtheorem{proposition}{Proposition}
\newtheorem{lemma}[theorem]{Lemma}
\def\x{{\mathbf x}}

\newcommand{\mat}[1]{\mathbf{#1}}
\newcommand{\norm}[1]{\left\lVert\mat{#1}\right\rVert}
\newcommand{\mscp}[2]{\big(\mat{#1}^{\top} \mat{#2}\big)}
\newcommand{\tr}[1]{\, \mathrm{tr} \{{#1}\}}

\title{Kernel computations from large-scale random features obtained by optical processing units
}

\author{\begin{tabular}{c}Ruben Ohana$^{\dag}$, Jonas Wacker$^\star$, Jonathan Dong$^{\dag,\ddagger}$, S\'ebastien Marmin$^\star$ \\
Florent Krzakala$^{\dag, \S}$, Maurizio Filippone$^\star$, Laurent Daudet$^\S$\end{tabular} \\
$\dag$ Laboratoire de Physique de l'Ecole Normale Sup\'erieure, ENS, Universit\'e PSL, CNRS \\ \normalsize Sorbonne Universit\'e, Universit\'e de Paris, F-75005 Paris, France \\ 
$\star$ Department of Data Science, EURECOM, 450 Route des Chappes, 06410 Biot, France \\ 
$\ddagger$ Laboratoire Kastler-Brossel, ENS, Universit\'e PSL, CNRS, Sorbonne Universit\'e, \\  \normalsize Coll\`ege de France, Universit\'e de Paris, F-75005 Paris, France \\ 
$\S$ LightOn, 2 rue de la Bourse, F-75002 Paris, France}

\begin{document}

\maketitle

\begin{abstract}
Approximating kernel functions with random features (RFs) has been a successful application of random projections for nonparametric estimation.
However, performing random projections presents computational challenges for large-scale problems.
Recently, a new optical hardware called Optical Processing Unit (OPU) has been developed for fast and energy-efficient computation of large-scale RFs in the analog domain. 
More specifically, the OPU performs the multiplication of input vectors by a large random matrix with complex-valued i.i.d. Gaussian entries, followed by the application of an element-wise squared absolute value operation -- this last nonlinearity being intrinsic to the sensing process. 
In this paper, we show that this operation results in a dot-product kernel that has connections to the polynomial kernel, and we extend this computation to arbitrary powers of the feature map.
Experiments demonstrate that the OPU kernel and its RF approximation achieve competitive performance in applications using kernel ridge regression and transfer learning for image classification.
Crucially, thanks to the use of the OPU, these results are obtained with time and energy savings.
\end{abstract}
\section{Introduction}

Kernel methods represent a successful class of Machine Learning models, achieving state-of-the-art performance on a variety of tasks with theoretical guarantees \cite{scholkopf2002learning,rudi2017falkon,caponnetto2007optimal}. 
Applying kernel methods to large-scale problems, however, poses computational challenges, and this has motivated a variety of contributions to develop them at scale; see, e.g., \cite{rudi2017falkon,smola2000sparse,zhang2013divide,NIPS2015_5936, EURECOM+5214}. 

Consider a supervised learning task, and let $\{\mat{x}_1, \ldots, \mat{x}_n\}$ be a set of $n$ inputs with $\mat{x}_i \in \mathbb{R}^d$ associated with a set of labels $\{t_1, \ldots, t_n\}$. 
In kernel methods, it is possible to establish a mapping between inputs and labels by first mapping the inputs to a high-dimensional (possibly infinite dimensional) Hilbert space $\mathcal{H}$ using a nonlinear feature map $\varphi:\mathbb{R}^d \rightarrow \mathcal{H}$, and then to apply the model to the transformed data. 
What characterizes these methods is that the mapping $\varphi(\cdot)$ does not need to be specified and can be implicitly defined by choosing a kernel function $k(\cdot, \cdot)$. 
While kernel methods offer a flexible class of models, they do not scale well with the number $n$ of data points in the training set, as one needs to store and perform algebraic operations with the kernel matrix $\mat{K}$, whose entries are $K_{ij} = k(\mat{x}_i, \mat{x}_j)$, and which require $\mathcal{O}(n^2)$ storage and $\mathcal{O}(n^3)$ operations.

In a series of celebrated papers~\cite{rahimi2008random,rahimi2009weighted}, Rahimi and Recht have proposed approximation techniques of the kernel function using random features (RFs), which are based on random projections of the original features followed by the application of a nonlinear transformation. 
In practice, the kernel function is approximated by means of the scalar product between finite-dimensional random maps $\phi:\mathbb{R}^d\rightarrow\mathbb{R}^D$:
\begin{equation}
\label{eq:random_features_approximation}
    k(\mat{x}_i,\mat{x}_j)=\langle\varphi(\mat{x}_i),\varphi(\mat{x}_j)\rangle_\mathcal{H} \approx \phi(\mat{x}_i)^\top \phi(\mat{x}_j) 
\end{equation}
The RF-based approximation turns a kernel-based model into a linear model with a new set of nonlinear features $\phi(\mat{x})$; 
as a result, 
the computational complexity is reduced from $\mathcal{O}(n^3)$ to $\mathcal{O}(n d D)$ to construct the random features and $\mathcal{O}(D^3)$ to optimize the linear model, where $D$ is the RF dimension and $n$ the number of data points. 
Furthermore, there is no need to allocate the kernel matrix, reducing the storage from $\mathcal{O}(n^2)$ to $\mathcal{O}(n D) + \mathcal{O}(D^2) $.
Unless approximation strategies to compute random features are used, e.g., \cite{le2013fastfood}, computing RFs is one of the main computational bottlenecks.

A completely different approach was pioneered in Saade {\it et al.}~\cite{saade_random_2016}, where the random projections are instead made via an analog optical device -- the  Optical Processing Unit (OPU) -- that performs these random projections literally at the speed of light and without having to store the random matrix in memory. 
Their results demonstrate that the OPU makes a significant contribution towards making kernel methods more practical for large-scale applications with the potential to drastically decrease computation time and memory, as well as power consumption. 
The OPU has also been applied to other frameworks like reservoir computing \cite{dong2018scaling, dong2019optical} and anomaly detection \cite{keriven2018newma}.

Building on the milestone work of \cite{saade_random_2016}, the goal of the present contribution is threefold: a) we derive in full generality the kernel to which the dot product computed by the OPU RFs converges, generalizing the earlier computation of \cite{saade_random_2016} to a larger class of kernels; b) we present new examples and a benchmark of applications for the kernel of the OPU; and c) we give a detailed comparison of the running time and energy consumption between the OPU and a last generation GPU.

\vspace{-0.09cm}

\section{The Optical Processing Unit}

\begin{figure}[h]
    \centering
    \includegraphics[width=0.4\textwidth]{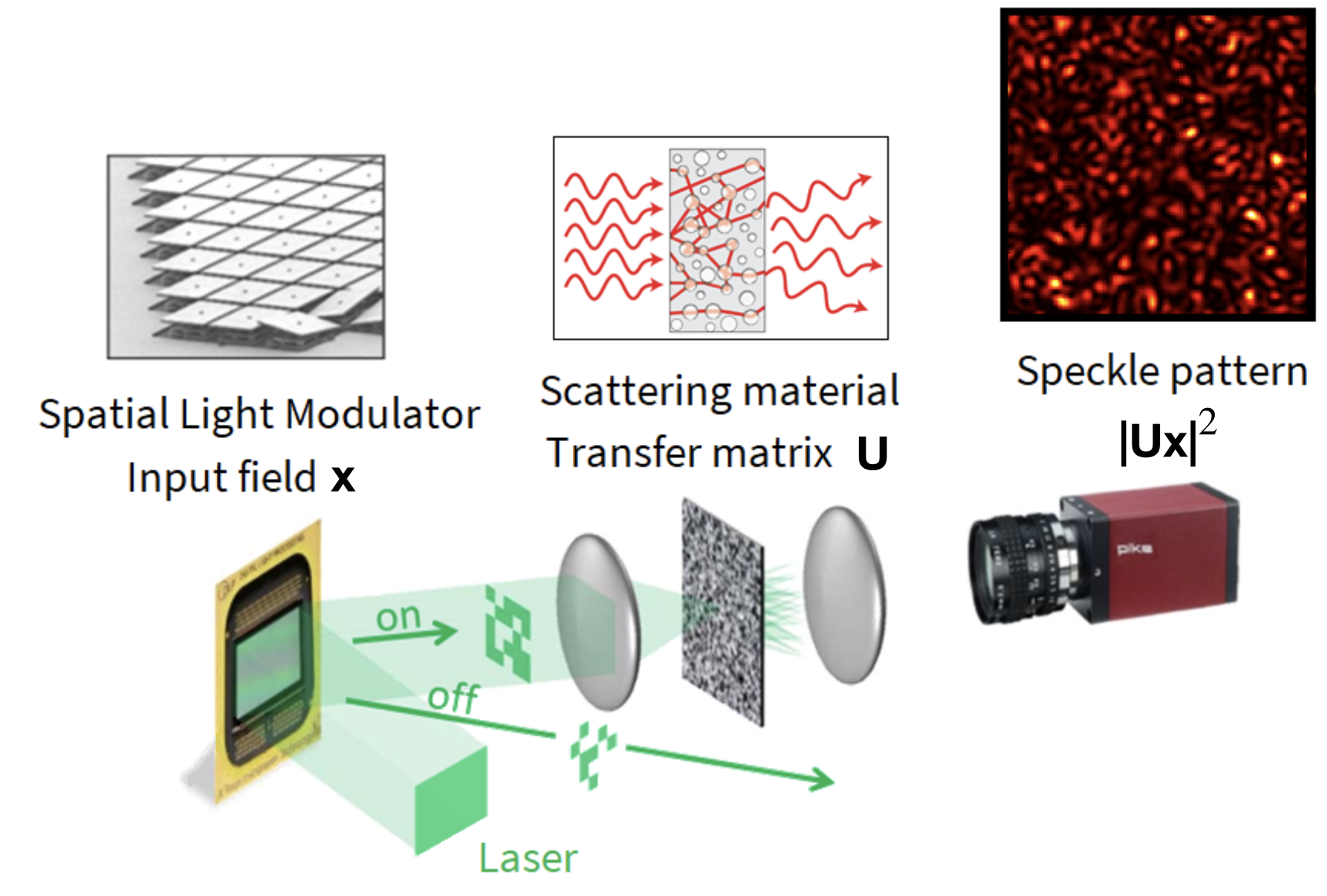}
    \caption{Experimental setup of the Optical Processing Unit (modified with permission from \cite{saade_random_2016}). The data vector is encoded in the coherent light from a laser using a DMD. Light then goes through a scattering medium and a speckle pattern is measured by a camera.}
    \label{fig:opu}
\end{figure}

The principle of the random projections performed by the Optical Processing Unit (OPU) is based on the use of a heterogeneous material that scatters the light that goes through it, see Fig.~\ref{fig:opu} for the experimental setup. The data vector $\mat{x} \in \mathbb{R}^d$ is encoded into light using a digital micromirror device (DMD). This encoded light then passes through the heterogeneous medium, performing the random matrix multiplication. As discussed  in \cite{liutkus2014imaging}, light going through the scattering medium follows many extremely complex paths, that depend on refractive index inhomogeneities at random positions. For a fixed scattering medium, the resulting process is still linear, deterministic, and reproducible. Reproducibility is important as all our data vectors need to be multiplied by the same realisation of the random matrix.

After going through the "random" medium, we observe a speckle figure on the camera, where the light intensity at each point is modelled by a sum of the components of $\mat{x}$ weighted by random coefficients. 
Measuring the intensity of light induces a non-linear transformation of this sum, leading to:
\pagebreak
\addtocounter{proposition}{-1}
\begin{proposition}
Given a data vector $\mat{x} \in \mathbb{R}^d$, the random feature map performed by the Optical Processing Unit is:
\begin{equation}
\label{eq:OPU-featuremap}
    \phi(\mat{x}) = \frac{1}{\sqrt{D}}|\mat{Ux}|^2
\end{equation}
where $\mat{U}\in\mathbb{C}^{D\times d}$ is a complex Gaussian random matrix whose elements $U_{i,j}\sim\mathcal{CN}(0,1)$, the variance being set to one without loss of generality, and depends on a multiplicative factor combining laser power and attenuation of the optical system. We will name these  RFs optical random features.
\end{proposition}


\section{Computing the Kernel}
When we map two data points $\mat{x},\mat{y}\in\mathbb{R}^d$ into a feature space of dimension $D$ using the optical RFs of Eq.~\ref{eq:OPU-featuremap}, we have to compute the following to obtain the associated kernel $k_2$:
\begin{align}
    k_2(\mat{x},\mat{y})&\approx \phi(\mat{x})^{\top} \phi(\mat{y}) 
    =\frac{1}{D}\sum_{j=1}^D |\mat{x}^{\top} \mat{u}^{(j)}|^2 |\mat{y}^{\top} \mat{u}^{(j)}|^2\\
    &\overset{D\to +\infty}{=}\int |\mat{x}^{\top} \mat{u}|^2 |\mat{y}^{\top} \mat{u}|^2\mu(\mat{u})\textrm{d}\mat{u}
    \label{eq: K2tendsto} 
\end{align}
\vspace{-0.1cm}
with $\mat{U}=[\mat{u}^{(1)},\ldots,\mat{u}^{(D)}]^{\top}$ and $\mat{u}^{(j)}\in\mathbb{R}^d, \forall j\in\{1,..,D\}$.
\begin{theorem}
\label{theo:expo2}
The kernel $k_2$ approximated by the dot product of optical random features of Eq.~\ref{eq:OPU-featuremap} is given by:
\begin{equation}
\label{eq:kernelexpo2}
    k_2(\mat{x},\mat{y}) = \norm{x}^2 \norm{y}^2 + \mscp{x}{y}^2
\end{equation}
where the norm is the $l_2$ norm.
\end{theorem}
\begin{proof}


By rotational invariance of the complex Gaussian distribution, we can fix $\mat{x} = \norm{x} \mat{e_1}$ and $\mat{y} = \norm{y}(\mat{e_1} \cos \theta + \mat{e_2} \sin\theta)$, with $\theta$ being the angle between $\mat{x}$ and $\mat{y}$, $\mat{e_1}$ and $\mat{e_2}$ being two orthonormal vectors. Letting $\mat{e_i}^\top \mat{u}=u_i\sim \mathcal{CN}(0,1)$, $i=1,2$
and $u_1^*$ be the complex conjugate of  $u_1$, we obtain:
\begin{align*}
    k_2(\mat{x},\mat{y}) &= \norm{x}^2 \norm{y}^2 \int|u_1|^2 |u_1\cos\theta+u_2\sin\theta|^2 \textrm{d}\mu(\mat{u}) \\
   &=\norm{x}^2 \norm{y}^2
    \int\bigg(|u_1|^4 \cos^2 \theta+|u_1|^2|u_2|^2 \sin^2 \theta + 2|u_1|^2\textrm{Re}(u_1^*u_2)\cos\theta\sin\theta\bigg)\textrm{d}\mu(u_1)\textrm{d}\mu(u_2)
\end{align*}
\noindent By a parity argument, the third term in the integral vanishes, and the remaining ones can be explicitly computed, yielding:
\begin{equation*}
    k_2(\mat{x},\mat{y}) = \norm{x}^2 \norm{y}^2(1+\cos^2\theta)=\norm{x}^2 \norm{y}^2 + \mscp{x}{y}^2
\end{equation*}
Two extended versions of the proof are presented in Appendix \ref{appendix:A}.
\end{proof}
Numerically, one can change the exponent of the feature map to $m\in\mathbb{R}^+$, which, using notations of Eq.~\ref{eq:OPU-featuremap}, becomes:
\begin{equation}
\label{eq:FMarbexpo}
    \phi (\mat{x}) = \frac{1}{\sqrt{D}}|\mat{Ux}|^m
\end{equation}
\pagebreak
\begin{theorem}
\label{th:kernel2k}
When the exponent $m$ is even, i.e. $m=2s$, $\forall s \in \mathbb{N}$, the dot product of feature maps of Eq.~\ref{eq:FMarbexpo} tends to the kernel $k_{2s}$ (for $D\rightarrow \infty$):
\begin{equation}
k_{2s}(\mat{x},\mat{y}) = \norm{x}^m\norm{y}^m \sum_{i=0}^s (s!)^2 {s \choose i}^2 \frac{\mscp{x}{y}^{2i}}{\norm{x}^{2i}\norm{y}^{2i}}
\label{eq:kernel2k}
\end{equation}
The proof is given in Appendix \ref{annex:proofevenexp}. Moreover, a generalization $\forall m \in\mathbb{R}^+$ can be established.
\end{theorem}
\noindent Eq.~\ref{eq:kernel2k} is connected to the polynomial kernel~\cite{scholkopf2002learning} defined as:
\begin{equation}
\label{eqn:bin-theorem}
(\nu + \mat{x}^{\top} \mat{y})^p
= \sum_{i=0}^{p} \binom{p}{i} \nu^{p-i} \mscp{x}{y}^i
\end{equation}
with $\nu \geq 0$ and $p \in \mathbb{N}$ the order of the kernel. For $\nu=0$ the kernel is called homogeneous. For $\nu > 0$ the polynomial kernel consists of a sum of lower order homogeneous polynomial kernels up to order $p$. It can be seen as having richer feature vectors including all lower-order kernel features.
For optical RFs raised to the power of $s \in \mathbb{N}$ we have a sum of homogeneous polynomial kernels taken to even orders up to $m=2s$.

Since $\mat{x}^{\top}\mat{y} =\norm{x}\norm{y} \cos{\theta}$,  the kernel scales with $\norm{x}^i\norm{y}^i$, which is characteristic to any homogeneous polynomial kernel. It is easy to extend this relation to the inhomogeneous polynomial kernel by appending a bias to the input vectors, i.e. $\mat{x}'^{\top}\mat{y}' = \nu + \mat{x}^{\top} \mat{y}$ when $\mat{x}'=(\sqrt{\nu}, \mat{x}_1, ..., \mat{x}_d)^{\top}$ and $\mat{y}'=(\sqrt{\nu}, \mat{y}_1, ..., \mat{y}_d)^{\top}$.
A practical drawback of this approach is that increasing the power of the optical RFs also increases their variance. Thus, convergence requires higher projection dimensions. Although high dimensional projections can be computed easily using the OPU, solving models on top of them poses other challenges that require special treatment~\cite{rudi2017falkon} (e.g. Ridge Regression scales cubically with $D$). Therefore, we did not include these cases in the experiments in the next section and leave them for future research.


\section{Experiments}

\begin{figure}[h]
    \centering
    \subfloat[$D \leq 10\,000$ (without optical RFs for $m=4$)]{{\includegraphics[width=0.5\textwidth, clip=true,trim=0 0 0 0]{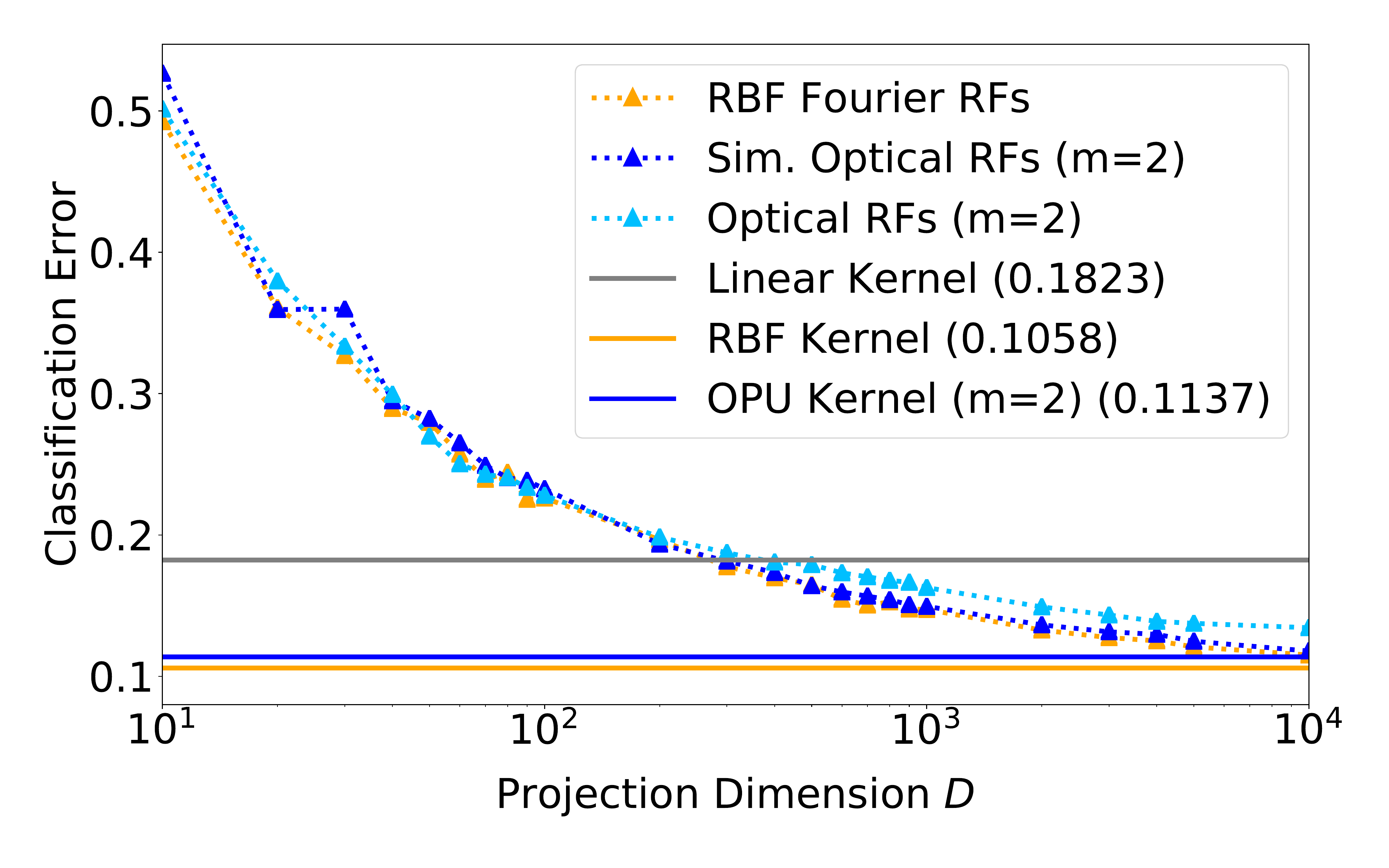} }}%
    \subfloat[$D \geq 10\,000$ (with optical RFs for $m=4$)]{{\includegraphics[width=0.5\textwidth, clip=true,trim=0 0 0 0]{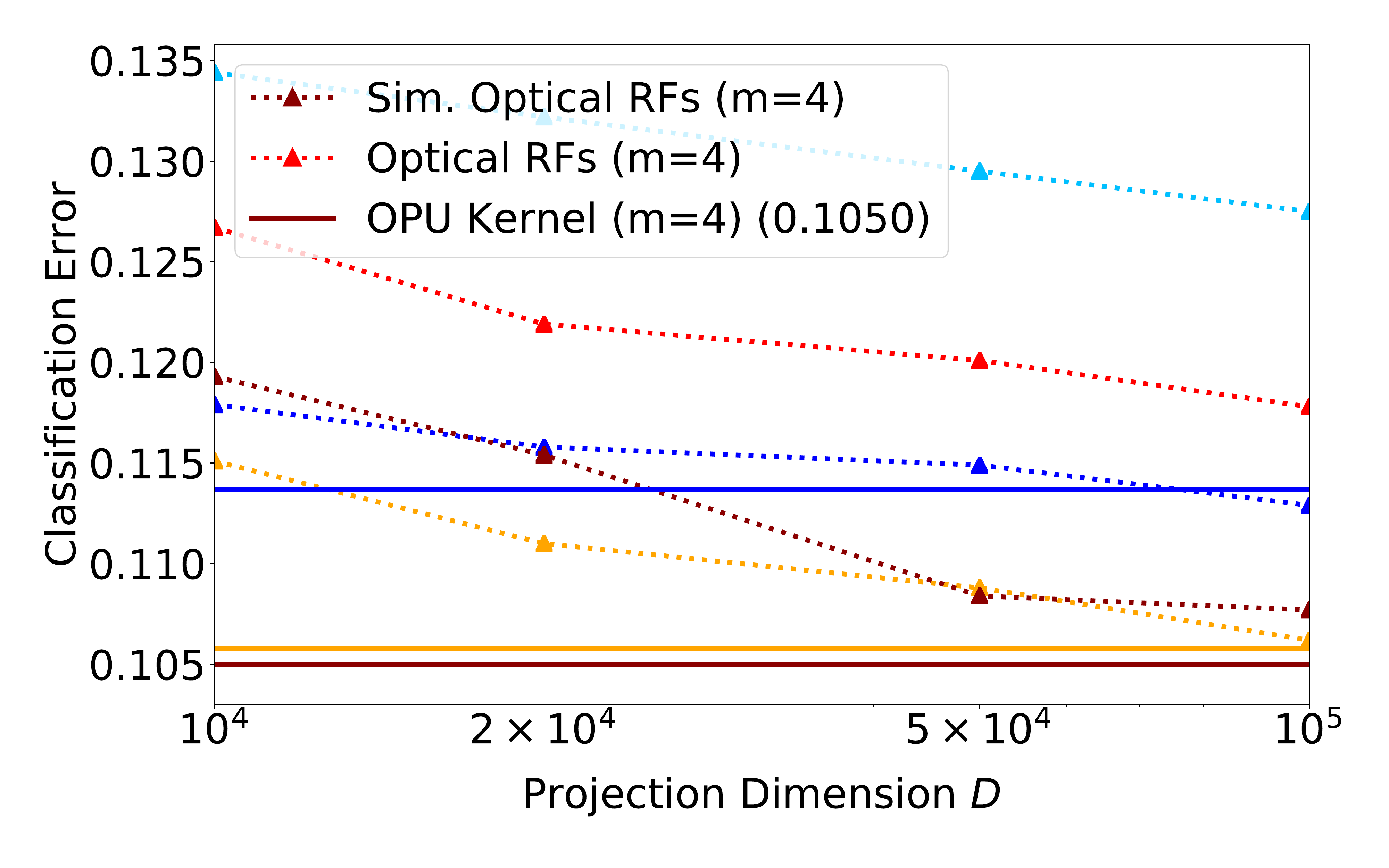} }}%
    \caption{Ridge Regression test error on Fashion MNIST for different RFs and projection dimensions $D$. Horizontal lines show the test error using the true kernel. Standard deviations for different seeds are negligibly small and not shown in the plot. Plot (a) compares optical RFs of degree $m=2$ to RBF Fourier RFs. Higher degree optical RFs are left out for better readability. The more slowly converging optical RFs for $m=4$ are added for larger $D$ in plot (b).}%
    \label{fig:fashion-mnist}
\end{figure}

In this section, we assess the usefulness of optical RFs for different settings and datasets. The model of our choice in each case is Ridge Regression. 
OPU experiments were performed remotely on the OPU prototype "Vulcain", running in the LightOn Cloud with library \texttt{LightOnOPU} v1.0.2. Since the current version only supports binary input data we decide to binarize inputs for all experiments using a threshold binarizer (see Appendix~\ref{appendix:C}). The code of the experiments is publicly available\footnote{ \href{https://github.com/joneswack/opu-kernel-experiments}{https://github.com/joneswack/opu-kernel-experiments}}.

\subsection{Optical random features for Fashion MNIST}

We compare optical RFs (simulated as well as physical) to an RBF Fourier Features baseline for different projection dimensions $D$ on Fashion MNIST. We use individually optimized hyperparameters for all RFs that are found for $D=10\,000$ using an extensive grid search on a held-out validation set. The same hyperparameters are also used for the precise kernel limit.
Fig.~\ref{fig:fashion-mnist} shows how the overall classification error decreases as $D$ increases. Part (b) shows that simulated optical RFs for $m=2$ and RBF Fourier RFs reach the respective kernel test score at $D=10^5$. Simulated optical RFs for $m=4$ converge more slowly but outperform $m=2$ features from $D=2\times10^4$. They perform similarly well as RBF Fourier RFs at $D=10^5$. The performance gap between $m=2$ and $m=4$ also increases for the real optical RFs with increasing $D$. This gap is larger than for the simulated optical RFs due to an increase in regularization for the $m=2$ features that was needed to add numerical stability when solving linear systems for large $D$.

The real OPU loses around 1.5\% accuracy for $m=2$ and 1.0\% for $m=4$ for $D=100\,000$, which is due slightly suboptimal hyperparameters to improve numerical stability for large dimensions. Moreover, there is a small additional loss due to the quantization of the analog signal when the OPU camera records the visual projection.


\subsection{Transfer learning on CIFAR-10}

\begin{table*}[h]
\centering
\setlength\tabcolsep{5.0pt}

\begin{tabularx}{\textwidth}{c|cccc|ccc|ccccc}
    
    Architecture & \multicolumn{4}{c|}{ResNet34} & \multicolumn{3}{c|}{AlexNet} & \multicolumn{5}{c}{VGG16} \\

    
    \hline
    
    Layer & L1 & L2 & L3 & Final &
    MP1 & MP2 & Final &
    MP2 & MP3 & MP4 & MP5 & Final \\
    
    Dimension $d$ & 4\,096 & 2\,048 & 1\,024 & 512 &
    576 & 192 & 9\,216 &
    8\,192 & 4\,096 & 2\,048 & 512 & 25\,088 \\

    \hline
    
    Sim. Opt. RFs & 30.4 & \textbf{24.7} & \textbf{28.9} & \textbf{11.6} & \textbf{38.1} & \textbf{41.9} & 19.6 & 28.2 & \textbf{20.5} & \textbf{20.7} & \textbf{29.8} & 15.2 (\textbf{12.9}) \\
    
    Optical RFs & 31.1 & 25.7 & 29.7 & 12.3 & 39.2 & 42.6 & 20.8 & 30.9 & 23.3 & 21.5 & 30.2 & 16.4  \\
    
    RBF Four. RFs & \textbf{30.1} & 25.2 & 30.0 & 12.3 & 39.4 & \textbf{41.9} & \textbf{19.1} & 28.0 & 20.7 & \textbf{20.7} & 30.1 & 14.8 (13.0) \\
    
    No RFs & 31.3 & 26.7 & 33.5 & 14.7 & 44.6 & 48.8 & 19.6 & \textbf{27.1} & 21.0 & 22.5 & 34.8 & \textbf{13.3} 
\end{tabularx}
\caption{Test errors (in \%) on CIFAR-10 using $D=10^4$ RFs for each kernel (except linear). Features were extracted from intermediate layers when using the original input size (32x32). Final convolutional layers were used with upscaled inputs (224x224). L(i) refers to the ith ResNet34 layer and MP(i) to the ith MaxPool layer of VGG16/AlexNet. Values for the kernel limit are shown in parenthesis (last column).}
\label{tab:cifar10}
\end{table*}
    
    



An interesting use case for the OPU is transfer learning for image classification. For this purpose we extract a diverse set of features from the CIFAR-10 image classification dataset using three different convolutional neural networks (ResNet34~\cite{He2015}, AlexNet~\cite{NIPS2012_4824} and VGG16~\cite{Simonyan14c}). The networks were pretrained on the well-known ImageNet classification benchmark~\cite{ILSVRC15}. For transfer learning, we can either fine-tune these networks and therefore the convolutional features to the data at hand, or we can directly apply a classifier on them assuming that they generalize well enough to the data. The latter case requires much less computational resources while still producing considerable performance gains over the use of the original features. This light-weight approach can be carried out on a CPU in a short amount of time where the classification error can be improved with RFs.

We compare Optical RFs and RBF Fourier RFs to a simple baseline that directly works with the provided convolutional features (no RFs).
Table~\ref{tab:cifar10} shows the test errors achieved on CIFAR-10. Each column corresponds to convolutional features extracted from a specific layer of one of the three networks.

Since the projection dimension $D=10^4$ was left constant throughout the experiments, it can be observed that RFs perform particularly well compared to a linear kernel when $D \gg d$ where $d$ is the input dimension. For the opposite case $D \ll d$ the lower dimensional projection leads to an increasing test error. This effect can be observed in particular in the last column where the test error of the RF approximation is higher than without RFs. The contrary can be achieved with large enough $D$ as indicated by the values for the true kernel in parenthesis.

A big drawback here is that the computation of sufficiently large dimensional random features may be very costly, especially when $d$ is large as well. This is a regime where the OPU outperforms CPU and GPU by a large margin (see Fig.~\ref{fig:power-consumption}) since its computation time is invariant to $d$ and $D$.

In general, the simulated as well as the physical optical RFs yield similar performances as the RBF Fourier RFs on the provided convolutional data.


\subsection{Projection time and energy consumption}

The main advantage of the OPU compared to a traditional CPU/GPU setup is that the OPU takes a constant time for computing RFs of arbitrary dimension $D$ (up to $D = 10^6$ on current hardware) for a single input. Moreover, its power consumption stays below 30 W independently of the workload.
Fig.~\ref{fig:power-consumption} shows the computation time and the energy consumption over time for GPU and OPU for different projection dimensions $D$.
In both cases, the time and energy spending do not include matrix building and loading. For the GPU, only the calls to the PyTorch function \texttt{torch.matmul} are measured and energy consumption is the integration over time of power values given by the \texttt{nvidia-smi} command.

For the OPU, the energy consumption is constant w.r.t. $D$ and equal to 45 Joules (30 W multiplied by 1.5 seconds). The GPU computation time and energy consumption are monotonically increasing except for an irregular energy development between $D=45\,000$ and $D=56\,000$. This exact irregularity was observed throughout all simulations we performed and can most likely be attributed to an optimization routine that the GPU carries out internally. The GPU consumes more than 10 times as much energy as the OPU for $D=58\,000$ (GPU memory limit).
The GPU starts to use more energy than the OPU from $D=18\,000$. The exact crossover points may change in future hardware versions. The relevant point we make here is that the OPU has a better scalability in $D$ with respect to computation time and energy consumption.



\begin{figure}[h]
\centerline{\includegraphics[width=1\linewidth, clip=true,trim=11 12 11 0]{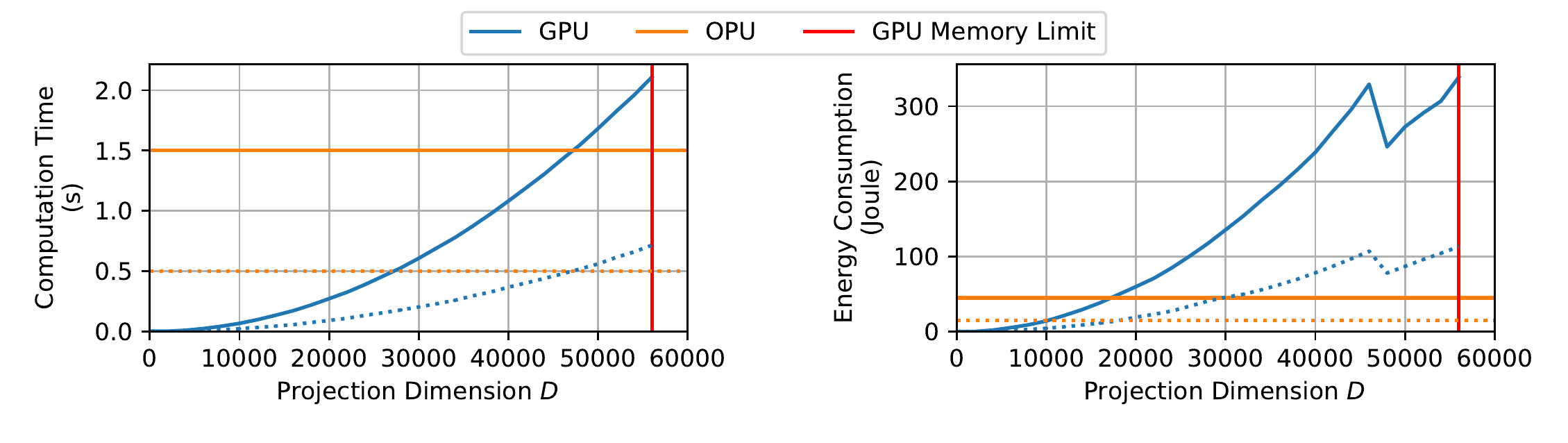}}
\caption{Time and energy spent for computing a matrix multiplication $(n,D)\times(D,D)$. The batchsize $n$ is $3000$ (solid line) or $1000$ (dotted). The curves cross each other at the same $D$ independent from $n$. We verified more precisely that time and energy are linear with $n$ for both OPU and GPU (experiments were run on a NVIDIA P100).}
\label{fig:power-consumption}
\end{figure}


\subsubsection*{Conclusion and perspectives}


The increasing size of available data and the benefit of working in high-dimensional spaces led to an emerging need for dedicated hardware. GPUs have been used with great success to accelerate algebraic computations for kernel methods and deep learning. Yet, they rely on finite memory, consume large amounts of energy and are very expensive.

In contrast, the OPU is a scalable memory-less hardware with reduced power consumption. In this paper, we showed that optical RFs are useful in their natural form and can be modified to yield more flexible kernels.
In the future, algorithms should be developed to deal with large-scale RFs, and other classes of kernels and applications should be obtained using optical RFs.

\subsubsection*{Acknowledgements}

RO acknowledges support by grants from R\'egion Ile-de-France.
MF acknowledges support from the AXA Research Fund and the Agence Nationale de la Recherche (grant ANR-18-CE46-0002). FK acknowledges support from  ANR-17-CE23-0023 and the Chaire CFM-ENS. The authors thank LightOn for access to the OPU and their kind support.


\clearpage
\label{sec:refs}
\bibliographystyle{IEEEbib}
\bibliography{refs}

\begin{thebibliography}{10}

\bibitem{scholkopf2002learning}
Bernhard Sch{\"o}lkopf, Alexander~J Smola, Francis Bach, et~al.,
\newblock {\em Learning with kernels: support vector machines, regularization,
  optimization, and beyond},
\newblock MIT press, 2002.

\bibitem{rudi2017falkon}
Alessandro Rudi, Luigi Carratino, and Lorenzo Rosasco,
\newblock ``Falkon: An optimal large scale kernel method,''
\newblock in {\em Advances in Neural Information Processing Systems}, 2017, pp.
  3888--3898.

\bibitem{caponnetto2007optimal}
Andrea Caponnetto and Ernesto De~Vito,
\newblock ``Optimal rates for the regularized least-squares algorithm,''
\newblock {\em Foundations of Computational Mathematics}, vol. 7, no. 3, pp.
  331--368, 2007.

\bibitem{smola2000sparse}
Alex~J Smola and Bernhard Sch{\"o}lkopf,
\newblock ``Sparse greedy matrix approximation for machine learning,''
\newblock 2000.

\bibitem{zhang2013divide}
Yuchen Zhang, John Duchi, and Martin Wainwright,
\newblock ``Divide and conquer kernel ridge regression,''
\newblock in {\em Conference on Learning Theory}, 2013, pp. 592--617.

\bibitem{NIPS2015_5936}
Alessandro Rudi, Raffaello Camoriano, and Lorenzo Rosasco,
\newblock ``Less is more: Nystr\"{o}m computational regularization,''
\newblock in {\em Advances in Neural Information Processing Systems 28},
  C.~Cortes, N.~D. Lawrence, D.~D. Lee, M.~Sugiyama, and R.~Garnett, Eds., pp.
  1657--1665. Curran Associates, Inc., 2015.

\bibitem{EURECOM+5214}
{K}urt {C}utajar, {E}dwin {B}onilla, {P}ietro {M}ichiardi, and {M}aurizio
  {F}ilippone,
\newblock ``{R}andom feature expansions for deep {G}aussian processes,''
\newblock in {\em {ICML} 2017, 34th {I}nternational {C}onference on {M}achine
  {L}earning, 6-11 {A}ugust 2017, {S}ydney, {A}ustralia}, {S}ydney,
  {AUSTRALIA}, 08 2017.

\bibitem{rahimi2008random}
Ali Rahimi and Benjamin Recht,
\newblock ``Random features for large-scale kernel machines,''
\newblock in {\em Advances in neural information processing systems}, 2008, pp.
  1177--1184.

\bibitem{rahimi2009weighted}
Ali Rahimi and Benjamin Recht,
\newblock ``Weighted sums of random kitchen sinks: Replacing minimization with
  randomization in learning,''
\newblock in {\em Advances in neural information processing systems}, 2009, pp.
  1313--1320.

\bibitem{le2013fastfood}
Quoc Le, Tam{\'a}s Sarl{\'o}s, and Alex Smola,
\newblock ``Fastfood-approximating kernel expansions in loglinear time,''
\newblock in {\em Proceedings of the international conference on machine
  learning}, 2013, vol.~85.

\bibitem{saade_random_2016}
A.~Saade, F.~Caltagirone, I.~Carron, L.~Daudet, A.~Dr\'emeau, S.~Gigan, and
  F.~Krzakala,
\newblock ``Random projections through multiple optical scattering:
  Approximating kernels at the speed of light,''
\newblock in {\em 2016 IEEE International Conference on Acoustics, Speech and
  Signal Processing (ICASSP)}, March 2016, pp. 6215--6219.

\bibitem{dong2018scaling}
Jonathan Dong, Sylvain Gigan, Florent Krzakala, and Gilles Wainrib,
\newblock ``Scaling up echo-state networks with multiple light scattering,''
\newblock in {\em 2018 IEEE Statistical Signal Processing Workshop (SSP)}.
  IEEE, 2018, pp. 448--452.

\bibitem{dong2019optical}
Jonathan Dong, Mushegh Rafayelyan, Florent Krzakala, and Sylvain Gigan,
\newblock ``Optical reservoir computing using multiple light scattering for
  chaotic systems prediction,''
\newblock {\em IEEE Journal of Selected Topics in Quantum Electronics}, vol.
  26, no. 1, pp. 1--12, 2019.

\bibitem{keriven2018newma}
Nicolas Keriven, Damien Garreau, and Iacopo Poli,
\newblock ``Newma: a new method for scalable model-free online change-point
  detection,''
\newblock {\em arXiv preprint arXiv:1805.08061}, 2018.

\bibitem{liutkus2014imaging}
Antoine Liutkus, David Martina, S{\'e}bastien Popoff, Gilles Chardon, Ori Katz,
  Geoffroy Lerosey, Sylvain Gigan, Laurent Daudet, and Igor Carron,
\newblock ``Imaging with nature: Compressive imaging using a multiply
  scattering medium,''
\newblock {\em Scientific reports}, vol. 4, pp. 5552, 2014.

\bibitem{He2015}
Kaiming He, Xiangyu Zhang, Shaoqing Ren, and Jian Sun,
\newblock ``Deep residual learning for image recognition,''
\newblock {\em arXiv preprint arXiv:1512.03385}, 2015.

\bibitem{NIPS2012_4824}
Alex Krizhevsky, Ilya Sutskever, and Geoffrey~E Hinton,
\newblock ``Imagenet classification with deep convolutional neural networks,''
\newblock in {\em Advances in Neural Information Processing Systems 25},
  F.~Pereira, C.~J.~C. Burges, L.~Bottou, and K.~Q. Weinberger, Eds., pp.
  1097--1105. Curran Associates, Inc., 2012.

\bibitem{Simonyan14c}
K.~Simonyan and A.~Zisserman,
\newblock ``Very deep convolutional networks for large-scale image
  recognition,''
\newblock {\em CoRR}, vol. abs/1409.1556, 2014.

\bibitem{ILSVRC15}
Olga Russakovsky, Jia Deng, Hao Su, Jonathan Krause, Sanjeev Satheesh, Sean Ma,
  Zhiheng Huang, Andrej Karpathy, Aditya Khosla, Michael Bernstein,
  Alexander~C. Berg, and Li~Fei-Fei,
\newblock ``{ImageNet Large Scale Visual Recognition Challenge},''
\newblock {\em International Journal of Computer Vision (IJCV)}, vol. 115, no.
  3, pp. 211--252, 2015.

\end{thebibliography}
\pagebreak
\appendix
\section{Extended proofs of Theorem \ref{theo:expo2}}
\subsection{Main proof of Theorem \ref{theo:expo2}}
\label{appendix:A}
\begin{proof}
Thanks to the rotational invariance of the complex gaussian random vectors $\mat{u}^{(k)}$, we can fix $\mat{x} = \|\mat{x}\| \mat{e_1}$ and $\mat{y} = \|\mat{y}\| (\cos\theta \mat{e_1} + \sin\theta \mat{e_2})$. $\theta$ represents the angle between $\mat{x}$ and $\mat{y}$ and $\cos(\theta) = \frac{\mat{x}^\top \mat{y}}{  \|\mat{x}\|\|\mat{y}\|}$. Let $\mat{e_i}^\top\mat{u}=u_i\sim\mathcal{CN}(0,1)$, $i=1,2$.

Thus the kernel function becomes:
\begin{equation}
    k_2(\mat{x}, \mat{y}) = \|\mat{x}\|^2 \|\mat{y}\|^2 \int |u_1|^2 |u_1 \cos\theta + u_2 \sin\theta|^2 \mu(\mat{u}) d\mat{u}
\end{equation}
We then expand the quadratic forms and compute the resulting gaussian integrals:
\begin{align*}
    \frac{1}{\|\mat{x}\|^2 \|\mat{y}\|^2} k_2(\mat{x},\mat{y}) 
    &= \int |u_1|^2 |u_1 \cos\theta + u_2 \sin\theta|^2 \frac{1}{\pi^d} e^{- \frac{\|\mat{u}\|^2}{2}} d\mat{u} \\
    &= \int |u_1|^2 |u_1 \cos\theta + u_2 \sin\theta|^2 \frac{1}{\pi^2} e^{- \frac{(|u_1|^2 + |u_2|^2)}{2}} d\mat{u} \\
    &= \int \left(|u_1|^4 \cos^2\theta + |u_1|^2 |u_2|^2 \sin^2\theta + 2 |u_1|^2 \text{Re}(\overline{u}_1 u_2) \cos\theta \sin\theta\right) \\
    &\qquad \frac{1}{\pi^2} e^{- \frac{|u_1|^2 + |u_2|^2}{2}} d\mat{u}
\end{align*}

The third term in the parenthesis is odd in $u_2$, so the integral of this term vanishes. Let's remark that if $u\sim\mathcal{CN}(0,\sigma^2)$ then $\mathcal{I}m(u),\mathcal{R}e(u) \sim \mathcal{N}(0,\sigma^{*2}=\frac{1}{2}\sigma^2)$. The two other terms can be computed easily using the moments of the gaussian distributions (the moment of order 2 of a complex gaussian random variable is $2\Gamma(2)\sigma^{*2}=2\sigma^{*2}$, the moment of order 4 is $2^2 \Gamma(3)\sigma^{*4}=8\sigma^{*4}$ where $\sigma^{*2}$
is the variance of the real and the imaginary part of $u_{1,2}$) .
\begin{align*}
    k_2(\mat{x}, \mat{y}) &= 8\sigma^{*4}\cos^2\theta+ 4\sigma^{*4}\sin^2\theta\\
    &= 4\sigma^{*4} \|\mat{x}\|^2 \|\mat{y}\|^2 (2 \cos^2\theta + \sin^2\theta) \\
    &= 4\sigma^{*4} \|\mat{x}\|^2 \|\mat{y}\|^2 (1 + \cos^2\theta) \\
    &= 4\sigma^{*4} \|\mat{x}\|^2 \|\mat{y}\|^2 + 4\sigma^{*4} \mscp{x}{y}^2\\
    &= \|\mat{x}\|^2 \|\mat{y}\|^2 (1 + \cos^2\theta) \quad\textrm{if}\quad\sigma^{*2}=\frac{1}{2}
\end{align*}
\end{proof}

\subsection{Alternative proof of Theorem~\ref{theo:expo2}}

The following is an alternative derivation of Theorem~\ref{theo:expo2} that breaks the complex random projection into its real and imaginary parts.

\begin{proof}
We can rewrite Equation~\ref{eq:OPU-featuremap} as:

\begin{equation}
\phi(\mat{x}) = \frac{1}{\sqrt{D}} | \mat{Ux} |^2
= \frac{1}{\sqrt{D}} \Big( (\mat{A x})^2 + (\mat{B x})^2 \Big),
\label{eq:OPUphi}
\end{equation}
where $\mat{A}, \mat{B} \in \mathbb{R}^{D \times d}$ are the real and imaginary parts of the complex matrix $\mat{U} \in \mathbb{C}^{D \times d}$. The elements of $\mat{A}$ and $\mat{B}$ are i.i.d. draws from a zero-centered Gaussian distribution with variance $\sigma^{*2}$.

Now we rewrite the kernel as:

\begin{align}
\label{eq:real_kernel}
k_2(\mat{x},\mat{y})
=~& \mathbb{E} \Big[{| \mat{U x} |^2 | \mat{U y}|^2 \Big] }
= \mathbb{E} \Big[
\Big( \mscp{a}{x}^2 + \mscp{b}{x}^2 \Big)
\Big( \mscp{a}{y}^2 + \mscp{b}{y}^2 \Big)
\Big]\nonumber\\
=~& \mathbb{E} \Big[ \mscp{a}{x}^2 \mscp{a}{y}^2 \Big] +
\mathbb{E} \Big[\mscp{b}{x}^2 \mscp{b}{y}^2 \Big] +
\mathbb{E} \Big[\mscp{a}{x}^2 \mscp{b}{y}^2\Big] +
\mathbb{E} \Big[\mscp{b}{x}^2 \mscp{a}{y}^2\Big]\nonumber\\
=~& 2 \Big( \underbrace{\mathbb{E} \Big[\mscp{a}{x}^2 \mscp{a}{y}^2\Big]}_{\substack{(1)}} + \underbrace{\mathbb{E} \Big[\mscp{a}{x}^2 \mscp{b}{y}^2\Big]}_{\substack{(2)}} \Big)
\end{align}

Term (1) in Equation~\ref{eq:real_kernel} can be seen as the expectation of the product of two quadratic forms $\mathbb{E} \big[ Q_1(\mat{a}) Q_2(\mat{a}) \big]$ where $Q_1(\mat{a}) = \mat{a^{\top} x x^{\top} a}$ and $Q_2(\mat{a}) = \mat{a^{\top} y y^{\top} a}$. Expectations of products of quadratic forms in normal random variables are well-studied by (Magnus, 1978)\footnote{The moments of products of quadratic forms in normal variables, Jan R. Magnus, statistica neerlandica, Vol. 32, 1978} and others.

Using their result, we can immediately solve Term (1):

\begin{equation}
\underbrace{\mathbb{E} \Big[\mscp{a}{x}^2 \mscp{a}{y}^2\Big]}_{\substack{(1)}}
= \sigma^{*4} \Big( \tr{ \mat{x x^{\top}} } \tr{ \mat{y y^\top} } + 2\tr{ \mat{x x^{\top} y y^{\top}} } \Big) = \sigma^{*4} \Big( \norm{x}^2 \norm{y}^2 + 2 \mscp{x}{y}^2 \Big)
\end{equation}

Term (2) is easy to solve:

\begin{equation}
    \underbrace{\mathbb{E} \Big[\mscp{a}{x}^2 \mscp{b}{y}^2\Big]}_{\substack{(2)}}
    = \mathbb{E} \Big[\mscp{a}{x}^2\Big] \mathbb{E} \Big[\mscp{b}{y}^2\Big]
    = \Big( \mat{x}^{\top} \mathbb{E} \Big[ \mat{a a^{\top}} \Big] \mat{x} \Big)
    \Big( \mat{y}^{\top} \mathbb{E} \Big[ \mat{b b^{\top}} \Big] \mat{y} \Big)
    = \sigma^{*4} \norm{x}^2 \norm{y}^2
\end{equation}

Inserting both terms into Equation~\ref{eq:real_kernel} yields the desired kernel equation.

\end{proof}

Although higher degree kernels can be derived in this manner as well, we proceed with the previous method for the following derivations.

\section{Proof of Theorem \ref{th:kernel2k}}
\subsection{Even exponents}
\label{annex:proofevenexp}
If we use the same reasoning as in the proof of Appendix \ref{appendix:A}, when the optical random feature is given as in Equation \ref{eq:FMarbexpo}, i.e. $\phi (\mat{x}) = \frac{1}{\sqrt{D}}|\mat{Ux}|^m$, we obtain:
\begin{align*}
    k_{2s}(\mat{x},\mat{y}) &= \int |\mat{x}^\top \mat{u}|^m |\mat{y}^\top \mat{u}|^m \mu(\mat{u})d\mat{u}\\
    &= \int \|\mat{x}\|^m\|\mat{y}\|^m|u_1|^m|u_1\cos\theta+u_2\sin\theta|^m\mu(\mat{u})d\mat{u}\\
    &= \|\mat{x}\|^m\|\mat{y}\|^m \int |u_1|^m|u_1\cos\theta+u_2\sin\theta|^m\mu(\mat{u})d\mat{u}\\
\end{align*}
Now we focus on the term $A = |u_1\cos\theta+u_2\sin\theta|^{2s}$ (with $m =2s$) as we focus on even powers: 

\begin{align*}
A &= (\overline{u}_1\cos\theta+\overline{u}_2\sin\theta)^s(u_1\cos\theta+u_2\sin\theta)^s\\
&= \bigg(\sum_{j=0}^s {s\choose j}(\overline{u}_1\cos\theta)^j(\overline{u}_2\sin\theta)^{s-j}\bigg)\bigg(\sum_{i=0}^s {s\choose i}(u_1\cos\theta)^i(u_2\sin\theta)^{s-i}\bigg)\\
&= \sum_i\sum_j{s\choose j}{s\choose i }\overline{u}_1^j u_1^i\overline{u}_2^{s-j}u_2^{s-i}(\cos\theta)^{i+j}(\sin\theta)^{2s-i-j}
\end{align*}
One can notice that $\mathbb{E}[A] = \sum_i\sum_j{s\choose j}{s\choose i }\mathbb{E}[\overline{u}_1^j u_1^i \overline{u}_2^{s-j} u_2^{s-i}]$ has its cross terms equal to 0: when $i\neq j$, the expectation inside the sum is equal to zero because $u_1$ and $u_2$ are i.i.d (their correlation is then equal to 0) or we can see this by rotational invariance (any complex random variable with a phase uniform between 0 and $2\pi$ has a mean equal to zero). Therefore we obtain:
\begin{equation*}
    \mathbb{E}[A] = \sum_{i=0}^s {s \choose i}^2 |u_1|^{2i}|\cos\theta|^{2i} |u_2|^{2(s-i)}|\sin\theta|^{2(s-i)}
\end{equation*}
Using the computation of $\mathbb{E}[A]$, we can therefore deduce the analytical formula for the kernel:
\begin{align*}
    \frac{k_{2s}(\mat{x},\mat{y})}{\|\mat{x}\|^m\|\mat{y}\|^m} &= \int|u_1|^{2s}\sum_{i=0}^s {s \choose i}^2 |u_1|^{2i}|\cos\theta|^{2i} |u_2|^{2(s-i)}|\sin\theta|^{2(s-i)}\mu(\mat{u})d\mat{u}\\
    &= \sum_{i=0}^s {s \choose i}^2 \mathbb{E}\bigg[|u_1|^{2(s+i)}|\cos\theta|^{2i}\bigg]\,\, \mathbb{E}\bigg[|u_2|^{2(s-i)}|\sin\theta|^{2(s-i)}\bigg]\\
    &= \sum_{i=0}^s {s \choose i}^2 |\cos\theta|^{2i}\,2^{s+i}\sigma^{*2(s+i)}\,\Gamma(s+i+1)\, |\sin\theta|^{2(s-i)}\,2^{s-i}\sigma^{*2(s-i)}\,\Gamma(s-i+1)\\
    &= \sum_{i=0}^s {s \choose i}^2 \,2^{2s}\sigma^{*4s}\,|\cos\theta|^{2i}\,|\sin\theta|^{2(s-i)} \Gamma(s+i+1)\,\Gamma(s-i+1)
\end{align*}

We can simplify this formula, by noticing that $\sin^2\theta=1-\cos^2\theta$. The latter formula becomes: 
\begin{align*}
    \frac{k_{2s}(\mat{x},\mat{y})}{\|\mat{x}\|^m\|\mat{y}\|^m} &= \sum_{i=0}^s {s \choose i}^2 |\cos\theta|^{2i}\,(1-\cos^2\theta)^{(s-i)}\,2^{2s}\sigma^{*4s}\,\Gamma(s+i+1)\,\Gamma(s-i+1)
\end{align*}
The new term can be expanded using a binomial expansion: $(1-\cos^2\theta)^{(s-i)}=\sum_{t=0}^{s-i}(-\cos^2\theta)^t$, leading to 
\begin{align*}
    \frac{k_{2s}(\mat{x},\mat{y})}{\|\mat{x}\|^m\|\mat{y}\|^m} &= \sum_{i=0}^s\sum_{t=0}^{s-i}{s \choose i}^2(s+i)!(s-i)!{s-i \choose t}(-1)^t\cos^{2(i+t)}\theta
\end{align*}
Using the change of variable $a=i+t$ and keeping $i=i$, we obtain:
\begin{align*}
    \frac{k_{2s}(\mat{x},\mat{y})}{\|\mat{x}\|^m\|\mat{y}\|^m} &= \sum_{a=0}^s\bigg[\sum_{i=0}^{a}{s \choose i}^2(s+i)!(s-i)!{s-i \choose a-i}(-1)^{a-i}\bigg]\cos^{2a}\theta
\end{align*}
We focus on the term between the brackets that we will call $T_a$:
\begin{align*}
    T_a &= \sum_{i=0}^{a}\frac{s!^2}{(s-i)!^2i!^2}(s+i)!(s-i)!\frac{(s-i)!}{(a-i)!(s-a)!}(-1)^{a-i}\\
    &=\sum_{i=0}^{a} (s!)^2 \frac{(s+i)!}{(i!)^2(a-i)!(s-a)!}(-1)^{a-i}\\
    &=(s!)^2{s \choose a}(-1)^a\sum_{i=0}^a {a \choose i}{s+i \choose i}(-1)^{i}
\end{align*}
Using the upper negation (${a\choose b}=(-1)^b{b-a-1\choose b}$) so here ${s+i\choose i} = {i - i -s-1 \choose i}(-1)^i= {-s-1\choose i}(-1)^i$ leads to
\begin{align}
    T_a &= (s!)^2{s \choose a}(-1)^a\sum_{i=0}^a {a \choose i}{-s-1 \choose i}\\
    &= (s!)^2{s \choose a}(-1)^a\sum_{i=0}^a {a \choose i}{-s-1 \choose -s-1-i}
\end{align}  
Now we can use the Vandermonde identity $\sum_{i=0}^n{a\choose i}{b \choose n-i} = {a+b \choose n}$, yielding:
\begin{align}
    T_a &= (s!)^2{s \choose a}(-1)^a {a-s-1\choose -s-1}=(s!)^2{s \choose a}(-1)^a {a-s-1\choose a}\\
        &=(s!)^2 {s\choose a} {s\choose a} =(s!)^2 {s\choose a}^2
\end{align}
where in the last line we used again the upper negation formula.\\ 
This leads to the desired result for $m=2s$
\begin{equation*}
    k_{2s}(\mat{x},\mat{y}) = \|\mat{x}\|^m\|\mat{y}\|^m \sum_{i=0}^s (s!)^2 {s \choose i}^2\cos^{2i}\theta
\end{equation*}

\section{Convergence properties}
For simplicity, we will consider random variables sampled from the normal distribution and not the complex normal one. We will be using the following lemma for proving the convergence in probability of the estimator of the kernel generated using optical random features toward the real kernel:
\begin{lemma}
\label{lemma:subexpo}
(Bernstein-type inequality)\footnote{Theorem 2.8.2 of  \textit{High-dimensional probability: An introduction with applications in data science}, Roman Vershynin. Vol. 47. Cambridge University Press, 2018.} Let $X_1,...,X_D$ be independent centered sub-exponential random variables, and let  $a = (a_1,...,a_D)\in\mathbb{R}^D$ . Then, for every $t\geq0$, we have: . 
\begin{equation}
    \mathbb{P}\bigg\{\bigg|\sum_{i=1}^Da_iX_i\bigg|\geq t\bigg\}\leq 2 \exp\bigg[-c\min\bigg(\frac{t^2}{K^2\|a\|_2^2},\frac{t}{K\|a\|_\infty}\bigg)\bigg]
\end{equation}
with $c$ an absolute constant and $K =\max_i\|X_i\|_{\psi_1}$ ($\|X_i\|_{\psi_1} = \sup_{p\geq1}p^{-1}(\mathbb{E}|X|^p)^{1/p}$ being the sub-exponential norm of $X_i$).\\

\end{lemma}
Let's start with the case when the exponent is $m=1$.\\ 
We know that $\mathbb{E}[\phi(\mat{x})\phi(\mat{y})]=k_1(\mat{x},\mat{y})$ so we can obtain centered random variables by doing:
\begin{align}
    \phi(\mat{x})^\top\phi(\mat{y})-k_1(\mat{x},\mat{y}) &=\phi(\mat{x})^\top\phi(\mat{y})-\mathbb{E}[\phi(\mat{x})^\top\phi(\mat{y})]\\
    &= \frac{1}{D}\sum_{i=1}^D\bigg[|\mat{x}^\top \mat{u}^{(i)}| \, |\mat{y}^\top \mat{u}^{(i)}|-k_1(\mat{x},\mat{y})\bigg]\\
    &:=\sum_{i=1}^Da_iX_i
\end{align}
where we have $a_i=\frac{1}{D}$ and $X_i:=|\mat{x}^\top \mat{V}^{(i)}\mat{y}|-k_1(\mat{x},\mat{y})$ (with $\mat{V}^{(i)}=\mat{u}^{(i)}\mat{u}^{(i)\top}$).\\
The matrix $\mat{U}^{(i)}$ has elements behaving sub-exponentially as it is the inner product of two sub-gaussian vectors. Moreover, we have $X_i=|\sum_{j,k}x_jV_{jk}^{(i)}y_k|-k_1(\mat{x},\mat{y})$, so by using lemma \ref{lemma:subexpo}, we can deduce that the sum of independent sub-exponential random variables (here the $V_{jk}^{(i)}$) is still sub-exponential. So we can conclude that $X_i$ is a centered sub-exponential random variable.\\
By applying the Bernstein-type inequality of lemma \ref{lemma:subexpo}, we obtain for the feature map of exponent 1:
\begin{align}
    \mathbb{P}\bigg\{\bigg|\phi(\mat{x})^\top\phi(\mat{y})-k_1(\mat{x},\mat{y})\bigg|\geq t\bigg\} &\leq 2 \exp\bigg[-c\min\bigg(\frac{t^2D}{K^2},\frac{tD}{K}\bigg)\bigg]
\end{align}
So the convergence of the estimator toward the kernel $k_1$ is sub-exponential.
\vspace{0.2cm}

Now for the more general case when $m>1$, we have $\phi(\mat{x})=\frac{1}{\sqrt{D}}|\mat{U}\mat{x}|^m$ so we can introduce the quantity $W_i(\mat{x},\mat{y}) = |\mat{x}^\top V^{(i)} \mat{y}|$ (which has a sub-exponential behavior) such that:
\begin{align}
    \mathbb{P}\bigg\{\frac{1}{D}\sum_{i=1}^DW_i(\mat{x},\mat{y})^m-k_m(\mat{x},\mat{y})\geq t\bigg\} &=\mathbb{P}\bigg\{\frac{1}{D}\sum_{i=1}^DW_i(\mat{x},\mat{y})^m\geq t +k_m(\mat{x},\mat{y})\bigg\}\\
    &=\mathbb{P}\bigg\{\sum_{i=1}^DW_i(\mat{x},\mat{y})^m\geq Dt + Dk_m(\mat{x},\mat{y})\bigg\}\\
    &\leq\mathbb{P}\bigg\{\sum_{i=1}^DW_i(\mat{x},\mat{y})^m\geq Dt + k_m(\mat{x},\mat{y})\bigg\} \quad(\textrm{as}\,\, D\geq1)\\
    &\leq D\,\mathbb{P}\bigg\{W_1(\mat{x},\mat{y})\geq (Dt + k_m(\mat{x},\mat{y}))^{1/m} \bigg\} \quad\textrm{(Union bound)}\\
    &\leq D\,\mathbb{P}\bigg\{W_1(\mat{x},\mat{y})\geq (Dt)^{1/m} + k_m(\mat{x},\mat{y})^{1/m} \bigg\} \quad\textrm{(as $m>1$)}\label{eq27}
\end{align}
Now we can use the following Lyapunov inequality. For $0<s<t$ we have:
\begin{equation*}
    (\mathbb{E}|X|^s)^{1/s}\leq (\mathbb{E}|X|^t)^{1/t}
\end{equation*}
However we know that $k_m(\mat{x},\mat{y})^{1/m} = \mathbb{E}[|\mat{x}^\top\mat{V}^{(i)}\mat{y}|^m]^{1/m}$ and $ k_1 = \mathbb{E}[|\mat{x}^\top\mat{V}^{(i)}\mat{y}|]$, $\forall i \in \{1,\ldots ,D\}$, so using Lyapunov inequality we can deduce that for $m>1$, $k_m(\mat{x},\mat{y})^{1/m}>k_1(\mat{x},\mat{y})$. \\
Using that result, we can bound Eq.~\ref{eq27} as follows:
\begin{align}
    \mathbb{P}\bigg\{\frac{1}{D}\sum_{i=1}^DW_i(\mat{x},\mat{y})^m-k_m(\mat{x},\mat{y})\geq t\bigg\}
    &\leq D\,\mathbb{P}\bigg\{W_1(\mat{x},\mat{y})\geq (Dt)^{1/m} + k_m(\mat{x},\mat{y})^{1/m} \bigg\}\\
    &\leq D\,\mathbb{P}\bigg\{W_1(\mat{x},\mat{y})\geq (Dt)^{1/m} + k_1(\mat{x,\mat{y)}} \bigg\}\\ &\leq D\,\mathbb{P}\bigg\{W_1(\mat{x},\mat{y})-k_1(\mat{x},\mat{y})\geq (Dt)^{1/m} \bigg\}
\end{align}
Using the fact that $W_1$ is sub-exponential and its expectation is $k_1$, then $W_1-k_1(\mat{x},\mat{y})$ is a centered random variable. It follows as a conclusion:
\begin{align}
    \mathbb{P}\bigg\{\frac{1}{D}\sum_{i=1}^D|\mat{x}^\top \mat{V}^{(i)} \mat{y}|^m-k_m(\mat{x},\mat{y})\geq t\bigg\} \leq D \, \exp(-C(Dt)^{1/m})
\end{align}
with C an positive absolute constant. It is not a tight bound but it gives us a behavior in the tails scaling as $\sim  \exp(-(Dt)^{1/m})$.\\
We can conclude for these convergence rates that the higher the exponent $m$ of the feature map is, the slower is the convergence of the estimator toward the real kernel. It has been noticed experimentally in  Fig.~\ref{fig:fashion-mnist} where the convergence of the estimator toward the real kernel for exponent $m=4$ is slower than the one for $m=2$. 

\section{Extended experimental description}
\label{appendix:C}

\subsection{Data encoding}
The current version of the OPU only supports binary inputs. There are different ways to obtain a binary encoding of the data. We apply a simple threshold binarizer to each feature $j$ of a datapoint $\mat{x}_i$:

\begin{equation}
    T(x_{ij}) = \begin{cases}
               1 \quad \text{if} \quad x_{ij} > \theta,\\
               0 \quad \text{otherwise.}
            \end{cases}
\end{equation}

The optimal threshold $\theta$ is determined for every dataset individually such that it maximizes the accuracy on a held-out validation set. Despite the drastic reduction from 32-bit floating point precision to 2 bits, the generalization error drops only by a small amount. The drop is bigger for the convolutional features than for the Fashion MNIST data.

\subsection{Hyperparameter search}
We carry out an extensive hyperparameter search for every dataset (including each convolutional feature set). The hyperparameters are optimized for every kernel individually using its random feature approximation at $D=10\,000$. A thorough hyperparameter search for every feature dimension as well as for the full kernels would be too expensive. Therefore, the hyperparameters are kept the same for different degrees of kernel approximation.

The optimal set of hyperparameters was found with a grid-search on a held-out validation set. For every kernel, we optimized the feature scale as well as the regularization strength alpha. For the linear and the OPU kernel a bias term that can be appended to the original features was added. For the RBF kernel the third hyperparameter dimension is the gamma/lengthscale parameter.

The scale and alpha parameters are optimized on a $\log_{10}$-scale that depends on each kernel. The gamma parameter is optimized on a $\log_2$-scale for Fashion MNIST; for each convolutional feature set it is found by trying out all values around the maximum non-zero decimal position of the heuristic $\gamma = 1 / d$ where $d$ is the original feature dimension. The bias parameter is determined on a $\log$-scale depending on the degree of the kernel.

Fig.~\ref{fig:grid-search} shows an example hyperparameter grid for simulated optical RFs of degree 2 applied to Fashion MNIST. Brighter colors correspond to higher validation scores.

\begin{figure}[h]
\centerline{\includegraphics[width=1\linewidth, clip=true,trim=0 0 0 0]{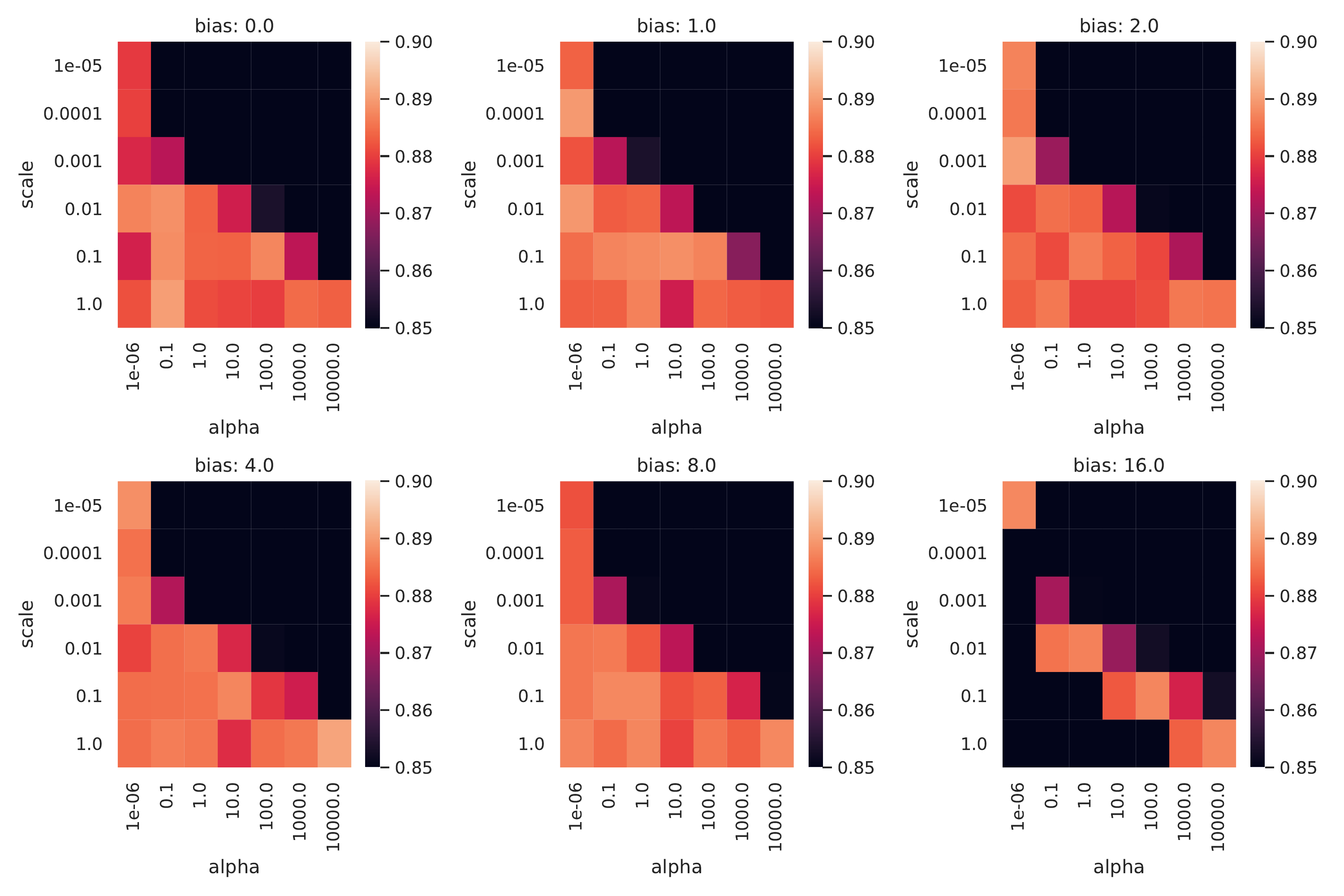}}
\caption{Example hyperparameter grid for simulated optical RFs of degree 2 applied to Fashion MNIST. It is easy to see that the upper triangular part of each grid matrix is a result of over-regularization ($\alpha$ is too large compared to the respective feature scale). Increasing the bias too much also leads to a degrading performance. The optimal hyperparameters are found in the lower-triangular part of one of the hyperparameter grids.}
\label{fig:grid-search}
\end{figure}

\subsection{Solvers for linear systems in Ridge Regression}

Ridge Regression can be solved either in its primal or dual form. In either case the solution is found by solving a linear system of equations of the form $\mat{Ax} = \mat{b}$. In the primal form we have $\mat{A}= \big( \phi(\mat{X})^{\top} \phi(\mat{X}) + \alpha \mat{I} \big)$ and $\mat{b} = \phi(\mat{X})^{\top} \mat{Y}$, whereas in the dual form $\mat{A} = \big(\mat{K} + \alpha \mat{I} \big)$ and $\mat{b} = \mat{Y}$. $\phi(\mat{X})$ is the random projection of the training data. $\mat{Y}$ are the one-hot encoded regression labels ($+1$ is used for the positive and $-1$ for the negative class). $\mat{K}$ is the n-by-n kernel matrix for the training data.

Solving the linear system has computational complexity $\mathcal{O}(D^3)$ or $\mathcal{O}(n^3)$ for the primal and dual form respectively. $D$ is the projection dimension and $n$ the number of datapoints.

Cholesky solvers turned out to work well for primal linear systems up to $D=10\,000$. For higher dimensions as well as for the dual form (computation of exact kernels) we used the conjugate gradients method.

For the computation of large matrix products as well as conjugate gradients, we developed a memory-efficient method that makes use of multiple GPUs in order to compute partial results. This way stochastic methods were avoided and exact solutions for the linear systems could be obtained.

One issue that arose during the experiments is that solving linear systems for optical RFs using GPUs worsened the conditioning of the matrix $\mat{A}$ due to numerical issues. A workaround was to increase $\alpha$ which led to slightly worse test errors.

In practice, we therefore recommend stochastic gradient descent based algorithms that optimize a quadratic regression loss. These allow to work with large-dimensional features while requiring much less memory and giving more numerical stability. Our method is only intended for theoretical comparison and should be used when exact results are needed.

\end{document}